\documentclass[11pt]{article}
\usepackage{rotating,multirow,amsmath,amssymb,fullpage}

\usepackage{natbib}
\usepackage{setspace}
\bibpunct[, ]{(}{)}{,}{a}{}{,}%

\newtheorem{definition}{Definition}
\newtheorem{theorem}{Theorem}
\newtheorem{lemma}{Lemma}

\newenvironment{proof}[1][Proof]{\textbf{#1.} }{\ \rule{0.5em}{0.5em}}
\begin{document}
	\title{Search and Rescue in the Face of Uncertain Threats}

	\author{
		Thomas Lidbetter$^{a}$\\
		{\em $^{a}$Department of Management Science and Information Systems,} \\
		{\em Rutgers Business School, Newark, NJ 07102, USA}\\
		{\em tlidbetter@business.rutgers.edu}
	}
	
	\maketitle
\onehalfspacing

\begin{abstract} \noindent
We consider a search problem in which one or more targets must be rescued by a search party, or {\em Searcher}. The targets may be survivors of some natural disaster, or prisoners held by an adversary. The targets are hidden among a finite set of locations, but when a location is searched, there is a known probability that the search will come to an end, perhaps because the Searcher becomes trapped herself, or is captured by the adversary. If this happens before all the targets have been recovered, then the rescue attempt is deemed a failure. The objective is to find the search that maximizes the probability of recovering all the targets. We present and solve a game theoretic model for this problem, by placing it in a more general framework that encompasses another game previously introduced by the author. We also consider an extension to the game in which the targets are hidden on the vertices of a graph. In the case that there is only one target, we give a solution of the game played on a tree.
\end{abstract}


\newpage 

\section{Introduction}

Many search and rescue operations may be dangerous for the search party, or {\em Searcher}. For example, searching in unstable buildings for earthquake survivors, searching for lost miners in a cave system, or performing a military rescue operation. One or more targets must be rescued, but there is a danger that the search is cut short due to one of these threats. For example, the Searcher is captured herself (in the case of a military operation), or to some incident caused by Nature results in the search being terminated (such as the Searcher becoming trapped herself).

We model such operations by introducing a new search model. We assume that a known number of targets are located among a finite set of possible locations, and at each location there is a given probability (which may depend on the location) that when that location is searched, the search will come to an end. The objective is to choose which order to search the locations to maximize the probability that the targets are all rescued.

We do not make any assumptions on the probability distribution with which the targets are located, but rather seek to find the randomized search that maximizes the probability of success in the worst case. Equivalently, we study a zero-sum game between the Searcher and a {\em Hider} who chooses where the targets are hidden. The latter way of framing the problem is particular appropriate for military applications. We call this game the {\em search and rescue game}, and it lies in the field of {\em search games}. For good general overviews of the literature on search games, see \cite{AG03}, \cite{Gal11} or \cite{H16}.

We begin by defining the game precisely in Section~\ref{sec:prelim}, in which we also point out a relation of the game to a scheduling problem introduced by \cite{Agnetis}. In Section~\ref{sec:gen}, we give a solution to the game. It turns out that the solution can be found by a similar method to that of the game introduced by the author in \cite{Lidbetter13}. The game studied in the latter paper also involves a known number of targets located among a finite number of locations, but there is no danger that the Searcher will be captured herself. Instead, there is a cost associated with searching each location, and the objective is to minimize the total cost of finding all the targets. Despite their similarities, the game of~\cite{Lidbetter13} and the game of this paper do not appear to be equivalent, so we unify them in a more general framework, simultaneously giving solutions to both.

In Section~\ref{sec:graph} we extend the game so that it is played on a graph. More precisely, we assume the hiding locations are vertices of the graph, and the Searcher must begin her search at a given vertex. We assume that the vertices of the graph must be searched according to an {\em expanding search}, a search paradigm introduced by the author in~\cite{AL13}. Roughly speaking, an expanding search of a graph is an ordering of the vertices such that each vertex is adjacent to some previously chosen vertex. In the context of the search and rescue game, this model of search is appropriate for situations in which locations must be searched contiguously, and after each one has been successfully searched, it can be marked as ``secure'', so that there is no danger of capture if the Searcher revisits it. The expanding search paradigm has been used more recently in \cite{ALD19} and \cite{AL18}.

We give a solution to the search and rescue game played on a tree in the case that there is only one target, by giving a recursive method for calculating optimal strategies and the value of the game. A similar approach was taken by \cite{Alpern10}, who found the solution to a different search game on a tree. In Alpern's game, the Searcher walks on a tree with the aim of minimizing the time taken to find a target, and the time taken to traverse an arc depends on the direction of travel.

This work takes a different approach to much of the literature on search games, which have the objective of minimizing some cost incurred in finding one or more target. In addition to the papers on search games already cited, we finish this section by briefly discussing some other recent work that takes a cost minimizing approach. 

The classic model of network search games, as studied in \cite{Gal79} and \cite{Gal01}, assumes that a Searcher, beginning at a fixed point of a network, wishes to find a immobile Hider located on the network in minimal time. This framework was extended in \cite{DG08} and \cite{ABG08} to allow an arbitrary starting point for the Searcher, and in~\cite{Alpern19} to restrict both the Searcher's starting point and the Hider's hiding point to a fixed subset of the network. ~\cite{AL15} considered a model of network search in which the Searcher has two speeds of travel: a slow speed at which she can detect the object when she passes it, and a fast speed at which she cannot. \cite{LidbetterEJOR} evaluated the performance of the Searcher strategy known as the {\em Random Chinese Postman Tour} in the classic network search model of \cite{Gal79}. \cite{BK13,BK15} consider a more general search game on a network where the vertices may have {\em search costs} that the Searcher has to pay to search them.

There are some instances of search games being considered that do not take a cost minimizing approach. For example \cite{LS16} consider a game in which a Searcher wishes to maximize the probability of finding a target before an unknown deadline; \cite{GC14} introduced a predator-prey game in which a predator wishes to maximize the probability of capturing the prey.

This is the first paper, as far as we are aware, to consider a search game that models a scenario in which the Searcher may be captured herself.

\section{Preliminaries} \label{sec:prelim}

We now formally define the search and rescue game. A set of $k$ targets must be rescued from a set of locations $S \equiv \{1,\ldots,n\}$, where $1 \le k \le n-1$.  The locations must be searched sequentially until all the targets have been found. If a location $i$ is searched, there is a probability $p_i \in (0,1)$ that the search is successful and all targets located there will be found. This probability is independent of where the location appears in the sequence. With probability $1-p_i$, the Searcher will be captured herself, and no more locations can be searched. Note that we disallow $p_i=0$ because the Hider could place a target in any such location $i$ so that the value of the game is $0$, and we disallow $p_i=1$ because it is trivially the case that such a location $i$ should be searched first.

Formally, a strategy for the Searcher is an ordering $\sigma:S \rightarrow S$, so that $\sigma(i)$ is the hiding location in the $i$th position in the ordering. We refer to an ordering $\sigma$ as a {\em search}. For the Hider, note that strategies in which more than one target is hidden in the same location are dominated, so we take the Hider's strategy set to be all subsets $H \in S^{(k)} \equiv \{A \subseteq S: |A|=k\}$. 

In order to define the payoff, first note that the probability the Searcher will be not be captured while searching a subset $A$ of hiding locations is 
\[
f(A) \equiv \prod_{i \in A}p_i.
\]
For a given Searcher strategy $\sigma$, let
\[
S^{\sigma}_i = \cup \{j \in S: \sigma^{-1}(j) \le i\}
\] 
be the first $i$ locations searched. Then for a given Hider strategy $H$ and Searcher strategy $\sigma$, the payoff $P(H,\sigma)$ of the game is $f(S^{\sigma}_i)$, where $i$ is minimal such that $H \subseteq S^{\sigma}_i$. That is, $P(H, \sigma)$ is the probability that the Searcher will not be captured by the time she finds all the targets. 

The Searcher's objective is to maximize the payoff, and the Hider's is to minimize it. This is a finite zero-sum game, so, by the minimax theorem of \cite{Neumann} for zero-sum games, it has optimal (max-min) mixed strategies and a value. A mixed strategy $s$ for the Searcher is a probability distribution over all searches of $S$, and a mixed strategy $h$ for the Hider is a probability distribution over $S^{(k)}$. For a mixed Hider strategy $h$ and a mixed Searcher strategy $s$, we denote the expected payoff of the game by $P(h,s)$.

\subsection{Relation to Unreliable Job Sequencing}

Consider the search and rescue game in the case that $k=1$. In this case, a mixed strategy for the Hider is a probability distribution over the set $S$ of locations. Such a strategy can be described by a vector of probabilities $\mathbf x \in \mathbb{R}^n$ with $\sum_i x_i = 1$, where $x_i$ is the probability that the Hider is in location $i$. Then for a given Searcher strategy $\sigma$, the probability $P(\mathbf x, \sigma)$ the target is rescued when the Hider uses some mixed strategy $\mathbf x$ is given by
\begin{align}
P(\mathbf x, \sigma) = x_{\sigma(1)}p_{\sigma(1)}+ x_{\sigma(2)}p_{\sigma(1)}p_{\sigma(2)}+ \cdots + x_{\sigma(n)} p_{\sigma(1)} \ldots p_{\sigma(n)}. \label{eq:scheduling}
\end{align}
We call the problem of choosing $\sigma$ to maximize $P(\mathbf x, \sigma)$ the {\em best response problem}, which is distinct from the problem of finding an optimal mixed strategy for the Searcher in the game. This problem has been considered by \cite{Agnetis} in the context of machine scheduling. Here, the problem is that $n$ jobs must be processed by a machine, and a reward of $x_i$ is obtained after successfully processing job $i$. The jobs are processed sequentially, and the probability that a job $i$ is successfully processed is $p_i$. Otherwise, with probability $1-p_i$, job $i$ fails, and no more jobs can be processed. The expected reward for a given ordering of the jobs is equal to $P(\mathbf x,\sigma)$.

Although \cite{Agnetis} consider the more general problem where the jobs are sequenced by multiple machines, they show that if there is only one machine, the optimal policy is given by ordering the jobs in non-increasing order of the index $p_ix_i/(1-p_i)$.

\cite{Agnetis} also point out a connection between their problem and a classic problem of \cite{Monma} in which $n$ components have to be sequentially tested until either a component fails or all the components pass the tests. The cost of testing component $i$ is $c_i$ and the probability it passes the test is $q_i$. \cite{Agnetis} show that with $p_i =q_i$ and $x_i = c_i/q_i$, this problem is equivalent to choosing an ordering to {\em minimize} the expression on the right-hand side of~(\ref{eq:scheduling}) (as opposed to their problem, which is to maximize it). The solution is to order $S$ in non-{\em decreasing} order of the index $x_ip_i/(1-p_i)$, rather than non-increasing order.

\section{A More General Search Game} \label{sec:gen}

We now place the search and rescue game in a more general context by defining a broader class of search games between a Hider and a Searcher. As before, the game is played on a set $S$ of locations, and this time $f:2^S \rightarrow \mathbb R$ is an arbitrary set function. We view $f$ as a reward function, and we assume that the values $f(A)$ are given by an oracle. The Hider's strategy set is all subsets $H \in S^{(k)}$, for some $k$, and the Searcher's strategy set is all searches (or orderings) of $S$. 

As before, for a given Hider strategy $H$ and Searcher strategy $\sigma$, the payoff $P=P_f$ of the game is given by $P(H, \sigma) = f(S^{\sigma}_i)$, where $i$ is minimal such that $H \subseteq S^{\sigma}_i$. The Searcher is the maximizer and the Hider the minimizer. We will denote this game by $\Gamma_f$.

We give a sufficient condition on $f$ for which $\Gamma_f$ has a simple closed-form solution. For $i \in S$ and $A \subseteq S$, let $f_A(i) = f(A \cup i) - f(A)$. (For brevity, we write $A \cup i$ for $A \cup \{i\}$ and $f(i)$ for $f(\{i\})$.)

\begin{definition} Let $f: 2^S \rightarrow \mathbb R$ be positive and $f(A)<f(B)$ for $B \subset A$. If there is a $\mathbf z \in \mathbb R^n$ with $z_r >0$ for all $r$, such that
\begin{align}
\frac{f_{A \cup j}(i)}{f_{A \cup i}(j)} = \frac{z_i}{z_j} \text{ when } i \notin A \text{ and } j \notin A, \label{eq:index}
\end{align}
then $f$ is {\em $\mathbf z$-indexable}.
\end{definition}

Note that the $f$ being strictly decreasing implies that $f_{A \cup j}(i) < 0$ for all $i,j \notin A$, so that the left-hand side of (\ref{eq:index}) is well-defined and positive.

The term ``indexable'' is inspired by the use of the term in \cite{Bertsimas} to describe dynamic and stochastic scheduling problems whose solution is given by assigning an index to each job and, at every stage, choosing the job with the highest index. 

We first observe that if $f$ is indexable, then for $k=1$, the best response problem for the $\Gamma_f$ is indexable, in the sense of \cite{Bertsimas}.

\begin{theorem} \label{thm:BR}
Suppose $f$ is $\mathbf z$-indexable, and consider a mixed Hider strategy $x \in \mathbb R^n$. The solution to the best response problem of $\Gamma_f$ is to search the elements of $S$ in non-increasing order of the index $x_i/z_i$.
\end{theorem}
\begin{proof}
The proof is a standard interchange argument. For a fixed search $\sigma$, the expected payoff of the game is
\[
P(x, \sigma) = \sum_{i=1}^n x_{\sigma(i)}f(S^{\sigma}_i).
\]
Suppose the elements of $S$ are searched according to some search $\sigma$ which is not in non-increasing order of the index $x_i/z_i$. Let $j$ be some element of $S$ such that the index of the element $i$ that is searched immediately after $j$ is larger than that of $j$. That is, $x_i/z_i > x_j/z_j$. Let $\sigma'$ be the same as $\sigma$, with the order of elements $i$ and $j$ transposed. 

Let $A = S^{\sigma}_{\sigma^{-1}(j)}-\{j\}$ be the subset of locations searched up but not including location $j$. Then
\begin{align*}
P(x, \sigma) - P(x, \sigma') &= (x_j f(A \cup j) + x_i f(A \cup \{i,j\})) - (x_i f(A \cup i) + x_j f(A \cup \{i,j\}))\\
& = x_i f_{A \cup i}(j) - x_j f_{A \cup j}(i) \\
& = z_j f_{A \cup j}(i) \left( \frac{x_i}{z_i} - \frac{x_j}{z_j} \right) \text{ (by~(\ref{eq:index}))} \\
& >0.
\end{align*}
Hence, $\sigma$ is not optimal, a contradiction. The theorem follows.
\end{proof}

\subsection{Examples of $\mathbf z$-Indexable Games} \label{sec:ex}

Here we describe some examples of games $\Gamma_f$ that are $\mathbf z$-indexable.

\subsubsection{The search and rescue game.} \label{sec:s&r}
It is easy to verify that for the search and rescue game, the function $f$ is $\mathbf z$-indexable. Indeed, for $i,j \notin A$, we have
\[
\frac{f_{A \cup j}(i)}{f_{A \cup i}(j)} = \frac{\left(p_i p_j\prod_{s \in A}p_s \right)-\left(p_j\prod_{s \in A}p_s \right)} {\left(p_i p_j\prod_{s \in A}p_s \right)-\left( p_i\prod_{s \in A}p_s \right)} = \frac{p_j(1-p_i)}{p_i(1-p_j)}.
\]
Thus, we can take $z_i = (1-p_i)/p_i$, and Theorem~\ref{thm:BR} implies that the solution to the best response problem for $k=1$ is to search the locations in non-increasing order of the index $x_i/z_i = x_ip_i/(1-p_i)$. This is consistent with the result of \cite{Agnetis}.

We may also extend the game, as in \cite{Agnetis}, by incorporating a discount factor on the probabilities $x_i$. More precisely, suppose that if there is a target at the location that appears in position $t$ of a search, there is a probability $\gamma^t$ that the target will be there when that location is searched, where $0 \le \gamma^t \le 1$. This could be the result of an increased likelihood that the target will not survive as time goes on, due to dangerous environmental factors or a malicious adversary. 

This is equivalent to the original game with new probabilities $p_i'=\gamma p_i$, and we obtain 
\[
\frac{f_{A \cup j}(i)}{f_{A \cup i}(j)} = \frac{p_j'(1-p_i')}{p_i'(1-p_j')} = \frac{p_j(1-\gamma p_i)}{p_i(1- \gamma p_j)}.
\]
Thus, we can take $z_i = (1-\gamma p_i)/p_i$, and it follows from Theorem~\ref{thm:BR} that the solution to the best response problem for $k=1$ is to search the locations in non-increasing order of $x_i/z_i = x_ip_i/(1- \gamma p_i)$, as shown directly in \cite{Agnetis}.

\subsubsection{An additive search game} \label{sec:additive}
This example is taken from \cite{Lidbetter13}. In this game, each location $i$ in $S$ has a {\em search cost} $c_i >0$ which the Searcher must pay to search it. The objective is to order the locations so as to minimize the sum of the search costs of all the locations searched until all $k$ targets have been found. In order to fit our framework of a decreasing function $f$, we take $f(A) =  \sum_{i \notin A} c_i$ to be the cost of locations {\em not} searched. In this case, for $i,j \notin A$, we have
\[
\frac{f_{A \cup j}(i)}{f_{A \cup i}(j)} = \frac{ ( - c_i -c_j + \sum_{s \notin A} c_s) - (-c_i - \sum_{s \notin A} c_s)}{(-c_i - c_j - \sum_{s \notin A} c_s) - (-c_j - \sum_{s \notin A} c_s)} = \frac{c_i}{c_j},
\]
and we can take $z_i = c_i$. Hence, by Theorem~\ref{thm:BR}, the solution to the best response problem for $k=1$ is to search the locations in non-increasing order of $x_i/z_i = x_i/c_i$. 

The best response problem to this game is a well-known search problem, first considered by~\cite{Bellman} (Chapter III, Exercise 3, p.90), and the solution is nothing new. It can equivalently be framed as the single machine scheduling problem posed by \cite{Smith}, now notated in the scheduling literature by $1||\sum w_j C_j$. The rule that the locations should be searched in non-increasing order of $x_i/z_i$ is known as {\em Smith's Rule}.

\subsubsection{A search game on a graph with traveling and search costs} \label{sec:network}
We describe here an application to a game introduced by~\cite{BK13}. The game is played on a graph with vertices $S$, and the Hider hides at one of the vertices. The Searcher starts at any vertex of her choosing and follows a walk in the network (that is, a sequence of vertices, each one of which is adjacent to the previous vertex). Each time the Searcher traverses an edge $e$, she pays a cost $d(e)$, and when visiting a vertex $i$, she can choose to search it for a search cost of $c_i$. If the Hider is located at a vertex, the Searcher finds him if and only if she pays the search cost. The payoff is the total cost to find the Hider.

\cite{BK13} solve the game for graphs that have a Hamiltonian, in the case that the traveling costs $d(e)$ are all equal to 1.  Here, we generalize their game to allow multiple targets to be hidden at vertices of the graph, so that the Searcher wants to minimize the cost of finding all the targets. In the case that the graph is complete (that is, there is an edge between every pair of vertices), the total cost of searching a subset $A$ of vertices is $ |A|-1 + \sum_{i \in A} c_i $. In this case, the game can be modeled by $\Gamma_f$, where we let $f(A) = |S-A| + \sum_{i \notin A} c_i $, to ensure the Searcher is the maximizer. Then for $i,j \notin A$, we have
\[
\frac{f_{A \cup j}(i)}{f_{A \cup i}(j)} = \frac{ ( |S-A|  - c_i -c_j + \sum_{s \notin A} c_s) - (|S-A| -c_j + \sum_{s \notin A} c_s)}{(|S-A| -c_i - c_j + \sum_{s \notin A} c_s) - (|S-A| -c_i + \sum_{s \notin A} c_s)} = \frac{1+c_i}{1+ c_j},
\]
so we can put $z_i = 1+c_i$.

Note that this game is actually equivalent to the game described in the previous subsection, if we take the cost of a location $i$ to be equal to $c_i + 1$.

\subsection{Solution to the Game $\Gamma_f$}

We now present a solution to the game $\Gamma_f$ when $f$ is $\mathbf z$-indexable. The solution mirrors the solution of the additive search game considered in \cite{Lidbetter13}. We will see that in the solution, the Searcher will randomize between mixed strategies which search some subset $A \in S^{(k)}$ first, in an arbitrary order, then choose the other elements of $S$ in a uniformly random order. We denote such a strategy by $s_A$. Note that the order of the first $k$ elements of the search do not matter, since the Searcher must search at least $k$ locations before finding all $k$ targets.
  
\begin{lemma} \label{lem:symm1}
Consider the game $\Gamma_f$, for an arbitrary function $f:2^S \rightarrow \mathbb R$. For any $A,B \in S^{(k)}$, we have $P_f(A, s_B) = P_f(B, s_A)$.
\end{lemma}
\begin{proof}
First suppose the Hider uses strategy $A$ and the Searcher uses strategy $s_B$, and we will calculate the expected payoff $P(A,s_B)$. Observe that the strategy $s_B$ must search all the elements of $B$ before finding the $k$ targets (since these are the first $k$ locations searched) and it must search all the elements of $A$ before finding the targets, because this is where they are located. Each other element of $S$, contained in the complement $(A \cup B)^c$ is searched with some probability $q$ before the $k$ targets are found. This probability $q$ must be the same for every element of $(A \cup B)^c$, since the elements of $B^c$ are searched in a uniformly random order.

Let $X$ be a random variable equal to $f(A\cup B \cup C)$, where the elements of $C$ are chosen independently with probability $q$ from the set $(A \cup B)^c$. Then, the expected payoff $P(A, s_B)$ is equal to the expectation $\mathbb E(X)$.

Clearly, by symmetry, the expected payoff $P(B, s_A)$ is also equal to $\mathbb E(X)$.
\end{proof}

We will use a simple but useful lemma, which we state without proof (see Lemma~2.6 of \cite{Lidbetter13} for the proof).
\begin{lemma} \label{lem:symm}
Consider an arbitrary zero-sum game with a symmetric payoff matrix in which Player I has a mixed strategy $x$ that makes Player II indifferent between all her pure strategies. Then $x$ is an optimal strategy for both Player I and Player II.
\end{lemma}

We can now solve the search and rescue game, but to state the solution we will use the following definition. 
\begin{definition} For a subset $A \subseteq S$, let 
\[
T_k(A) \equiv \sum_{B \in A^{(k)}} \prod_{i \in B} z_i.
\]
\end{definition}
Note that 
\begin{align}
T_{k}(A \cup i)z_j - T_{k}(A \cup j)z_i = z_j T_{k}(A) - z_i T_{k}(A). \label{eq:T}
\end{align}

\begin{theorem} \label{thm:main}
Consider the search game $\Gamma_f$, and suppose $f$ is $\mathbf z$-indexable. Then it is optimal for the Hider to use the mixed strategy $q$, whereby a set $A \in S^{(k)}$ is chosen with probability $q_A$ given by  
\[
q_A \equiv \frac{\prod_{i \in A} z_i}{T_k(S)},
\]
It is optimal for the Searcher to use the strategy $s$ that chooses $s_A$ with probability proportional to $q_A$. 
\end{theorem}
\begin{proof}
Suppose the Hider uses the strategy $q$ described in the statement of the theorem. We will show transposing any two adjacent elements of a given search $\sigma$ leaves the expected payoff unchanged. Since any search can be obtained from $\sigma$ by a sequence of such transpositions, this is sufficient to prove that all searches have the same expected cost against this Hider strategy.

Thus, suppose in a search $\sigma$, the element $i$ comes immediately before $j$, and let $\sigma'$ be the search that is the same as $\sigma$ except that elements $i$ and $j$ are transposed. If $j$ comes in position $k$ or earlier in $\sigma$, then clearly transposing $i$ and $j$ leaves the expected payoff unchanged. Otherwise, we can compute the difference in the expected payoffs of the two searches against the Hider strategy $q$ as follows. Let $A$ denote the set of locations searched before $i$.
\begin{align}
P(q, \sigma) - P(q, \sigma') &= \left( \frac{T_{k-1}(A)z_i}{T_k(S)} f(A \cup i) + \frac{T_{k-1}(A \cup i)z_j}{T_k(S)}f(A \cup \{i,j\}) \right) \nonumber \\
& \quad -  \left( \frac{T_{k-1}(A)z_j}{T_k(S)} f(A \cup j) + \frac{T_{k-1}(A \cup j)z_i}{T_k(S)}f(A \cup \{i,j\}) \right). \label{eq:diff}
\end{align}
Considering the coefficient of $f(A \cup \{i,j\})$ in~(\ref{eq:diff}), and using~(\ref{eq:T}) with $k$ replaced by $k-1$, we get
\begin{align*}
P(q, \sigma) - P(q, \sigma') &= \frac{T_{k-1}(A)}{T_{k}(S)} \left( z_j f_{A \cup j}(i) - z_i  f_{A \cup i}(j) \right) \\
&=0,
\end{align*}
by~(\ref{eq:index}).

The argument above shows that the value of the game is at most $V$, where $V=P(q,\sigma)$ is the payoff when the Hider uses the strategy $q$ against any Searcher strategy $\sigma$. If we restrict the Searcher to strategies of the form $s_A$, then the value of the restricted game is equal to $V$ and the strategies $s$ and $q$ are optimal, by Lemmas~\ref{lem:symm1} and~\ref{lem:symm}. Since the value, $V$ of the restricted game is a lower bound for the value of the original (unrestricted) game, the strategies $q$ and $s$ must also be optimal in the original game.
\end{proof}

Theorem~2.1 of \cite{Lidbetter13}, which gives the solution of the additive search game, follows as a corollary of Theorem~\ref{thm:main} of this paper. Our theorem also gives a solution to the search and rescue game. If $k=1$, there is a particularly simple expression for the value of the game.

\begin{theorem}
In the search and rescue game it is optimal for the Hider to choose a subset $A \in S^{(k)}$ with probability 
\[
q_A \equiv \lambda_k \prod_{i \in A}\frac{1-p_i}{p_i}, \text{ where } \lambda_k \equiv \left( \sum_{B \in S^{(k)}} \prod_{i \in B} \frac{1-p_i}{p_i} \right)^{-1}.
\]
It is optimal for the Searcher to choose a subset $A\in S^{(k)}$ of locations to search first with probability $q_A$, then search the remaining locations in a uniformly random order. For $k=1$, the value $V$ of the game is 
\begin{align}
V \equiv \lambda_1(1 - \prod_{i \in S} p_i) \label{eq:value}
\end{align}
\end{theorem}
\begin{proof}
The optimality of the Hider's and Searcher's strategies follows immediately from Theorem~\ref{thm:main}. To prove the correctness of~(\ref{eq:value}), it is sufficient to show that against the given Hider strategy, some (and therefore, every) search has expected payoff $V$. Let $\sigma$ be the search that chooses the locations in increasing order from $1$ to $n$. Then
\begin{align*}
P(q, \sigma) &= \sum_{i=1}^n q_{i} \prod_{j \le i} p_j  \\
& = \lambda_1 \sum_{i=1}^n \frac{1-p_i}{p_i} \prod_{j \le i} p_j  \\
& = \lambda_1 \sum_{i=1}^n (1-p_i) \prod_{j <i} p_j.
\end{align*}
The sum is telescopic, and reduces to the expression on the right-hand side of~(\ref{eq:value}).
\end{proof}

\section{The Search and Rescue Game on a Graph} \label{sec:graph}

In this section we consider the search and rescue game played on a graph. That is, we assume that the set of hiding locations $S$ are the vertices $V(G)$ of a graph $G$ with edge set $E(G)$, which is a collection of unordered pairs of vertices $(i,j), i \neq j$. The Searcher begins at a specified vertex $O$ of the graph, called the {\em root}, and can search the graph using an {\em expanding search}, which is a search paradigm introduced by the author in~\cite{AL13}. An expanding search of a graph $G$ with root $O$ is an ordering of the vertices $S$, starting with $O$, such that each vertex in the ordering is adjacent to some previous vertex. 

Formally, a Searcher strategy is a ordering $\sigma$ of $S$ such that $\sigma(1)=O$ and if $i>1$, then $(\sigma(i) , \sigma(i-1)) \in E(G)$. In this section we will refer to such an ordering simply as a {\em search}. If $\alpha$ is the restriction of some search $\sigma$ to some set $ \{i,i+1,\ldots,j\}$, then we call $\alpha$ a {\em subsearch} of the vertices $A = \sigma(\{i,i+1,\ldots,j\})$. In other words, $\alpha$ describes the sequence of vertices in positions $i$ through $j$ in the search.

The only difference between the game played on a graph and the original version of the game described in Section~\ref{sec:prelim} is that the Searcher's strategy set is restricted. In particular, a Hider strategy is still a probability distribution on $S$, given by a vector $\mathbf x \in \mathbb R^n$; the probability that the Searcher is not captured when she searches a vertex $v$ is denoted by $p_v$; the payoff of the game is calculated in the same way. We relax the restriction $p_v \in (0,1)$ slightly, by allowing $p_v$ to take the value $1$ as long as $v \neq O$ is not a leaf (a vertex of degree 1), otherwise it could be removed from the graph without changing the value of the game. The reason for this is that it will be convenient later to make the assumption that all vertices have degree at most $3$, and we will justify this assumption adding vertices $v$ with $p_v=1$ to the tree to obtain an equivalent game on a tree with the desired property.

Playing the game on a network complicates things, and we therefore restrict ourselves to the case $k=1$, leaving larger values of $k$ for future work.

We illustrate expanding search with an example. Consider the tree in depicted in Figure~\ref{fig:tree}. The vertices are labeled with letters, and the probabilities $p_i$ are shown next to the vertices. One possible expanding search on this network visits the vertices in the order $O, D, A, B, C$. If the target is located at $B$, say, then the payoff is $(1/2) \cdot (3/5)\cdot (2/3) \cdot (1/3) = 1/15$.

\begin{figure}[h]
\begin{center}
\includegraphics[scale=0.5]{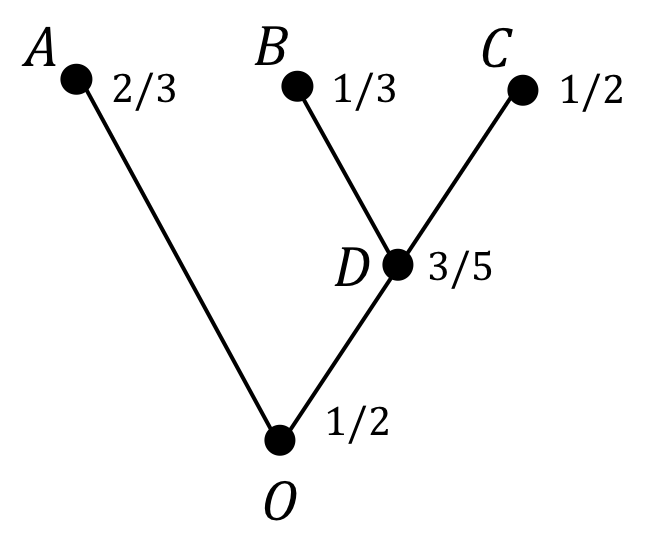}
\caption{A tree with probabilities $p_v$ indicated.}
\label{fig:tree}
\end{center}
\end{figure}

Note that if $p_O=1$ and the graph is a star (that is, a tree where $O$ is the only vertex of degree greater than 1), then an expanding search of the graph corresponds simply to an (unrestricted) ordering of all the leaf vertices. The search and rescue game played on such a graph is therefore equivalent to the game without any network structure. Equivalently, the game could be played on a complete graph.

%\begin{definition}
%We say that the cost function $f: 2^S \rightarrow \mathbb R$ is {\em decomposable} if there exists some positive function $g: 2^S \rightarrow \mathbb R$ such that
%\begin{align}
%f_B(A) = g(B)f(A), \label{eq:decomp}
%\end{align}
%for any $B \subseteq S$.
%\end{definition}
%Note that the concept of decomposable is more general than the concept of $\mathbf z$-indexable. Indeed, if $f$ is decomposable, then for $i,j \in S$,
%\[
%\frac{f_{A \cup j}(i)}{f_{A \cup i}(j)} = \frac{g(A \cup j)f(i)}{g(A \cup i)f(j)},
%\]
%so we may take $z_i = f(i)/g(A \cup i)$. 

%Also note that all the examples of cost functions from Section~\ref{sec:ex} are decomposable....

In order to solve the search and rescue game played on a tree, we define something similar to the index $z_i$ in Section~\ref{sec:gen}. First we introduced some more notation. For a subset $A\subseteq S$, let $\pi(A)$ denote $\prod_{i\in A} p_i$. If $G$ is a graph, we may write $\pi(G)$ to express $\pi(V(G))$. We also expand the definition of the payoff function $P$. For a subsearch $\alpha$ of vertices $\sigma(\{i, i+1,\ldots, j\})$, and a fixed Hider strategy $\mathbf x$, let 
\[
P(\mathbf x, \alpha) = x_{\alpha(i)} p_{\alpha(i)} + x_{\alpha(i+1)}p_{\alpha(i)} p_{\alpha(i+1)}+ \cdots + x_{\alpha(j)}p_{\alpha(i)}\cdots p_{\alpha(j)}.
\]
In particular, if $\sigma$ is a search of the whole of $S$, then $P(\mathbf x, \sigma)$ is the payoff of the game when the Hider uses $\mathbf x$. We will usually drop the $\mathbf x$ in $P(\mathbf x, \alpha)$, and simply write $P(\alpha)$ when there is no ambiguity.

We remark that the solution of the best response problem for the game played on a tree follows from the work of~\cite{Monma}, since the function $P$ can be easily shown to satisfy what the authors call the {\em series-network decomposition property}. We give the solution here to the {\em game} played on a tree.

For a fixed Hider strategy $\mathbf x$ and subsearch $\alpha$ of vertices $A$ with $\pi(A) \neq 1$, we define an index 
\[
  I(\alpha) \equiv I_{\mathbf x}(\alpha) \equiv P(\alpha)/(1-\pi(A)).
\]
The restriction $\pi(A) \neq 1$ insures that $I(\alpha)$ is well defined. Note that if $\alpha$ consists of a single element $i$, then the index of $\alpha$ is equal to $x_i p_i/(1-p_i)$, which is the same as the index determining the optimal search in the best response problem for the game played with no network structure. We now prove a more general lemma, which says that if two subsearches are disjoint and can be executed consecutively in some order, then the subsearch with the highest index should come first.

\begin{lemma} \label{lem:density}
Let $\mathbf x$ be some fixed Hider strategy and let $\sigma$ be a search. Suppose some subsearch $\alpha$ of $\sigma$ searches a subset $A \subseteq S$ of vertices immediately before some other subsearch $\beta$ searches a subset $B \subseteq S$ disjoint from $A$, with $\pi(A), \pi(B) \neq 1$. Let $\sigma'$ be the same as $\sigma$ except that the order of $\alpha$ and $\beta$ are transposed. Then $P(\sigma) \le P(\sigma')$ if and only if
\begin{align}
I(\alpha) \ge I(\beta), \label{eq:density}
\end{align}
with $P(\sigma) = P(\sigma')$ if and only if (\ref{eq:density}) holds with equality. 
\end{lemma}
\begin{proof}
Let $C \subseteq S$ be the set of locations searched immediately before $\alpha$ in $\sigma$. Then 
\begin{align*}
P(\sigma') - P(\sigma) &= (\pi(C)P(\alpha) + \pi(C \cup A) P(\beta)) - (\pi(C)P(\beta) + \pi(C \cup B)P(\alpha)) \\
&= \pi(C) (1-\pi(A))(1-\pi(B)) \left(I(\alpha) - I(\beta) \right),
\end{align*}
by definition of $I(\alpha)$ and $I(\beta)$. The lemma follows immediately.
\end{proof}

Lemma~\ref{lem:density} is a variation of the {\em Search Density Lemma}, found in various guises in studies of other search games, for example in \cite{Alpern10}, \cite{AL13} and \cite{FLV19}.

We present a solution to the game on a tree. We may assume that the maximum degree of any vertex of the tree is $3$. If not, then by successively adding vertices $v$ with $p(v) = 1$, we can iteratively transform the tree into a tree with degree at most $3$ such that the value of the game is the same on both trees and there is a one-to-one correspondence between strategies on one and on the other. Thus, any solution of the game for the transformed tree can be mapped back onto the original tree. For example, suppose Figure~\ref{fig:tree2} depicts a subgraph of a particular tree: in particular, the degree $4$ vertex $A$ and its neighbors. Assume that $B$ is closer to $O$ than $A$.  In other words, the path from $A$ to $O$ contains the vertex $B$, so that in any search of the tree, $B$ must appear before $A$. Then the subgraph depicted on the left of the figure can be replaced with the subtree shown on the right, where $p_X=1$, so that there is zero probability that the Searcher will be captured after she searches the vertex $X$. Strategies on the new tree map onto strategies on the original one in a natural way, with no alteration to the payoffs.

\begin{figure}[h]
\begin{center}
\includegraphics[scale=0.5]{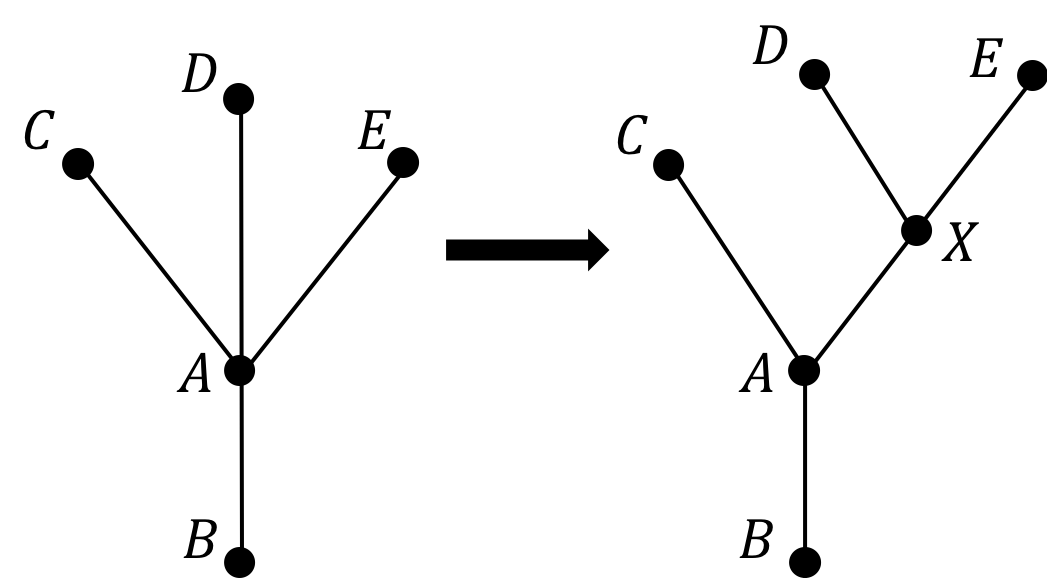}
\caption{Transformation of a degree 4 vertex.}
\label{fig:tree2}
\end{center}
\end{figure}

We recursively define a strategy $h \equiv h_G$ for the Hider on a tree $G$ which we will later prove is optimal. We also define recursively a quantity $V_G$, which we will later prove is the value of the game. We first define the {\em branches} of a tree $G$ with root $O$ as the connected components of the tree obtained when $O$ and its incident edges are removed from $G$. The roots of the branches are the neighbors of $O$. We call vertices $v$ with degree 3 {\em branch vertices}. We also refer to $O$ as a branch vertex if it has degree 2.

\begin{definition}[{\textbf Tree hiding strategy}] The Hider strategy $h_G$ is defined recursively as follows for rooted trees $G$. If $G$ has only has one vertex $v$, then let $h_G(v)=1$ and $V_G=p_v$. If $G$ has more than one vertex, there are two cases, depending on whether or not $O$ is a branch vertex.

\textbf{Case 1.} The root $O$ is not a branch vertex. In this case, let $G'$ be the unique branch of $G$ and let $O'$ be its root. The Hider strategy $h_G$ on $G$ is given by $h_G(O)=0$ and $h_G(v)=h_{G'}(v)$ for all vertices $v$ in $G'$. Let $V_G = p_O V_{G'}$.

\textbf{Case 2.} 
The root $O$ is a branch vertex. In this case, let $G_1$ and $G_2$ be the branches of $G$ and let $O_1$ and $O_2$ be their roots. Then we set
\[
h_G(G_i)= \lambda_G \left( \frac{1-\pi(G_i)}{V_{G_i}} \right), i=1,2,
\]
where
\[
\lambda_G = \left(  \frac{1-\pi(G_1)}{V_{G_1}} +  \frac{1-\pi(G_2)}{V_{G_2}} \right)^{-1}
\]
is a normalizing factor. Then for a vertex $v$ in $G_i$, we set 
\[
h_G(v) = h_G(G_i) h_{G_i}(v).
\]
We also set
\[
V_G = p_O \lambda_G (1-\pi(G_1)\pi(G_2)).
\]

\end{definition}

Note that the support of $h$ is the set of leaves of $G$ (excluding $O$, if it is a leaf). It is obvious that the support of any optimal strategy must be the set of leaves, because all other vertices are dominated by some leaf. 

We illustrate the computation of $h_G$ and $V_G$ for the tree depicted in Figure~\ref{fig:tree}. For any vertex $v$, let $G(v)$ be the subtree of $G$ containing all vertices whose path to $O$ contains $v$. Using $V_{G(B)} = 1/3$ and $V_{G(C)} = 1/2$, we compute
\begin{align*}
\lambda_{G(D)} &= \left(  \frac{1-\pi(G(B))}{V_{G(B)}} +  \frac{1-\pi(G(C))}{V_{G(C)}} \right)^{-1} \\
&= \left(  \frac{1-1/3}{1/3} +  \frac{1-1/2}{1/2} \right)^{-1}  = 1/3.
\end{align*}
Hence, we find that
\begin{align*}
h_{G(D)}(B) &= \lambda_{G(D)} \left( \frac{1-\pi(G(B))}{V_{G(B)}} \right) = (1/3) \left( \frac{1-1/3}{1/3} \right) = 2/3 \text{ and}\\
h_{G(D)}(C) &= \lambda_{G(D)} \left( \frac{1-\pi(G(C))}{V_{G(C)}} \right) = (1/3)\left( \frac{1-1/2}{1/2} \right) = 1/3.
\end{align*}
Also, 
\[
V_{G(D)} = p_D \lambda_{G(D)} (1- \pi(G(B))\pi(G(C))) = (3/5)(1/3)(1 - (1/3)(1/2)) = 1/6.
\]
Now using $V_{G(A)} = 2/3$, we compute
\begin{align*}
\lambda_{G} &= \left(  \frac{1-\pi(G(A))}{V_{G(A)}} +  \frac{1-\pi(G(D))}{V_{G(D)}} \right)^{-1} \\
&= \left(  \frac{1-2/3}{2/3} +  \frac{1-(1/3)(1/2)(3/5)}{1/6} \right)^{-1}  = 10/59.
\end{align*}
Hence,
\begin{align*}
h_{G}(A) &= \lambda_{G} \left( \frac{1-\pi(G(A))}{V_{G(A)}} \right) = (10/59) \left( \frac{1-2/3}{2/3} \right) = 5/59 \text{ and}\\
h_{G}(G(D)) &= \lambda_{G} \left( \frac{1-\pi(G(D))}{V_{G(D)}} \right) = (10/59)\left( \frac{1-(1/3)(1/2)(3/5)}{1/6} \right) = 54/59.
\end{align*}
It follows that
\begin{align*}
h_G(B) &= h_G(G(D)) h_{G(D)}(B) = (54/59)(2/3) = 36/59 \text{ and} \\
h_G(C) &= h_G(G(D)) h_{G(D)}(C) = (54/59)(1/3) = 18/59.
\end{align*}
The value of the game is
\[
V_G = p_O \lambda_{G} (1- \pi(G(A))\pi(G(D))) = (1/2) (10/59)(1- (2/3)(1/3)(1/2)(3/5)) = 14/177.
\]
This completes the calculation of the tree hiding strategy for the tree in Figure~\ref{fig:tree}.

To show that the Hider strategy $h_G$ guarantees an expected payoff of at most $V_G$, we first show that any depth-first search of $G$ has expected payoff $V_G$ against $h_G$, where the formal definition of a depth-first search is as follows.
\begin{definition}[\textbf{Depth-first search}]
Let $G$ be a tree with root $O$.  A depth-first search of $G$ is a search such that for any vertex $v$, all the other vertices of $G(v)$ appear in the search in some order immediately after $v$.
\end{definition}

\begin{lemma} \label{lem:value}
Suppose the Hider is located on a tree $G$ according to the tree hiding strategy. Then any depth-first search $\sigma$ of $G$ has expected payoff $P(\sigma) = V_G$.
\end{lemma}
\begin{proof}
Let $\sigma$ be a depth-first search of $G$. The proof is by induction on the number of vertices of $G$. The lemma is trivially true when there is only one vertex, so suppose $G$ has at least~2 vertices. There are two cases, depending on whether or not $O$ is a branch vertex. First suppose $O$ is not a branch vertex, in which case let $G'$ be its one branch, with root $O'$. By the induction hypothesis, the expected payoff of any depth-first search of $G'$ is $V_G$. Then clearly the expected payoff of $\sigma$ is
\[
P(\sigma) = p(O) V_{G'} \equiv V_G,
\]
by definition of $V_G$. 

In the other case, $O$ is a branch vertex and let $G_1$ and $G_2$ be the branches of $G$, with roots $O_1$ and $O_2$, respectively. By the induction hypothesis, any depth-first search of $G_i$ has expected payoff $V_{G_i}$ for $i=1,2$. Suppose, without loss of generality, that $\sigma$ performs successive depth-first searches $\sigma_1$ and $\sigma_2$ of $G_1$ and $G_2$ in that order, after searching $O$. Then the expected payoff of $\sigma$ is
\[
P(\sigma)  = p_O(P(\sigma_1)  +  \pi(G_1) P(\sigma_2) ).
\]
By induction, $P(\sigma_i) = h_{G}(G_i) V_{G_i}$ for all $i=1,2$, and by definition of $h$, we have $h_G(G_i) = \lambda_G (1-\pi(G_i))/V_{G_i}$. Hence,
\[
P(\sigma)  = p_O \lambda_G (  (1-\pi(G_1) + \pi(G_1)(1-\pi(G_2))) = p_O \lambda_G( 1 - \pi(G)) \equiv V_G.
\]
\end{proof}

In order to show that the value of the game is bounded above by $V_G$, it is sufficient to prove that depth-first searches are best responses to the Hider strategy $h_G$. This follows from the fact that the indices of both branches of a branch vertex are equal. We prove both of these next.

\begin{lemma} \label{lem:DF}
Suppose the Hider is located on a tree $G$ according to the strategy $h_G$. 
\begin{enumerate}
\item[(i)] Suppose $v$ is a branch vertex and let $\sigma_1$ and $\sigma_2$ be depth-first searches of the branches $G_1$ and $G_2$ of $G(v)$. Then $I(\sigma_1) = I(\sigma_2)$.
\item[(ii)] Any depth-first search is a best response to the tree hiding strategy $h_G$, and $h_G$ ensures the expected payoff of the game is at most~$V_G$.
\end{enumerate}
\end{lemma}
\begin{proof}
For part (i), by Lemma~\ref{lem:value}, we have that $P(\sigma_i) = h_G(G_i) V_{G_i} =h_G(G(v)) h_{G(v)}(G_i) V_{G_i}$ for $i=1,2$. By definition of $h_{G(v)}(G_i)$, we have
\[
P(\sigma_i) = h_{G(v)} \lambda_{G(v)} (1- \pi(G_i)).
\]
Hence,
\[
I(\sigma_i) = \frac{P(\sigma)}{1- \pi(G_i)} = h_{G(v)} \lambda_{G(v)}.
\]
This expression is independent of $i$.

For part (ii), suppose there exists a best response $\sigma$ to $h$ that is not depth-first. Then there must exist branch vertex $v$ such that the two branches $G_1$ and $G_2$ of $G(v)$ are not searched consecutively, but the subsearches of the branches are both depth-first. Without loss of generality, assume that $G_1$ is searched before $G_2$. We may also assume that $v$ appears in the search immediately before $G_1$, since if some other vertex $w$ were searched immediately before $G_1$, then the order of search of $v$ and $w$ could be swapped, and the expected payoff would not be any greater.

So there must exist consecutive subsearches $\sigma_1,\tau,\sigma_2$, where $\sigma_1$ and $\sigma_2$ are depth-first searches of $G_1$ and $G_2$ and $\tau$ is a subsearch of some other subset $A$ of vertices, disjoint from $G(v)$.

Since $\pi$ is optimal, by Lemma~\ref{lem:density}, we must have 
\begin{align}
I(\sigma_1) \ge I(\tau) \ge I(\sigma_2) \label{eq1}
\end{align}
But, by part (i), we have $I(\sigma_1)= I(\sigma_2)$, and it follows that all the inequalities in~(\ref{eq1}) hold with equality.

Therefore, by Lemma~\ref{lem:density}, the search $\sigma'$ that results from the subsearches $\sigma_1$ and $\tau$ being swapped has the same expected payoff as $\sigma$, and is therefore a best response to $h_G$. 

But $\sigma'$ searches $G(v)$ in a depth-first manner, and applying this argument repeatedly, we can transform $\sigma$ into a depth-first search which is also a best response to $h_G$. By Lemma~\ref{lem:value}, every depth-first search has the same expected payoff, and is therefore a best response to $h_G$, and the Hider can ensure that the value of the game is at most $V_G$.
\end{proof}

We now define the strategy mixed strategy $s = s_G$ that will turn out to be optimal for the Searcher. Similarly to the Hider's strategy, we define it recursively. We first need a technical lemma to ensure that the strategy $s_G$ is well defined.

\begin{lemma} \label{lem:Vlb}
For any tree $G$, we have that $1 \ge V_G \ge \pi(G)$.
\end{lemma}
\begin{proof}
We prove this by induction on the number of vertices of $G$. It is clearly true when $G$ has only one vertex, so suppose $G$ has more than one vertex and that the lemma is true for trees with fewer vertices than $G$. 

First suppose that $O$ is not a branch vertex and let $O'$ be the neighbor of $O$. Then by the induction hypothesis, $1\ge V_{G'} \ge \pi(G')$, so $V_G \equiv p(O) V_{G'} \ge p(O) \pi(G') = \pi(G)$. Also, clearly $V_G \le p(O) \le 1$.

Now suppose that $O$ is a branch vertex, and let $G_1$ and $G_2$ be the two branches. Then by the induction hypothesis, $V_{G_1} \ge \pi(G_1)$ and $V_{G_2} \ge \pi(G_2)$, and it follows that
\begin{align*}
\lambda(G) &\ge \left( \frac{1- \pi(G_1)}{\pi(G_1)} + \frac{1-\pi(G_2)}{\pi(G_2)} \right)^{-1} \\
& =   \frac{\pi(G_1)\pi(G_2)}{ \pi(G_2)(1- \pi(G_1)) + \pi(G_1)(1 - \pi(G_2))} \\
& \ge \frac{\pi(G_1)\pi(G_2)}{ 1- \pi(G_1) + \pi(G_1)(1 - \pi(G_2)) } \\
& = \frac{\pi(G_1)\pi(G_2)}{ 1 - \pi(G_1) \pi(G_2) }. 
\end{align*}
It follows that $V_G \equiv p_O \lambda(G)(1- \pi(G_1) \pi(G_2)) = \pi(G)$.

Also, by the induction hypothesis, $V_{G_1} \le 1$ and $V_{G_2} \le 1$, and it follows from the definition of $\lambda(G)$ that
\begin{align*}
\lambda(G) &\le (1- \pi(G_1)+ 1- \pi(G_2) )^{-1} \\
& \le (1 - \pi(G_1) + \pi(G_1)(1 - \pi(G_2)))^{-1} \\
& = (1 - \pi(G_1)\pi(G_2))^{-1}.
\end{align*}
If follows that $V_G \equiv p_O \lambda(G)(1- \pi(G_1)\pi(G_2)) \le 1$.
\end{proof}

\begin{definition}[Tree searching strategy] The tree searching strategy $s_G$ is a probabilistic choice of depth-first searches of a tree $G$, and is fully described by specifying which branch is searched first at each branch vertex. Suppose $v$ is such a vertex, and let $G_1$ and $G_2$ be the two branches. Then $G_1$ is searched first with probability
\[
q_{G_1} =  \lambda(G) \left(\frac{1}{V_{G_1}} - \frac{\pi(G_2)}{V_{G_2}} \right), %\quad q_2 =  \lambda(G) \left(\frac{1}{V_{G_2}} - \frac{\pi(G_1)}{V_{G_1}} \right).
\]
otherwise $G_2$ is searched first.
\end{definition}
It is easy to check that $q_{G_1}+q_{G_2} =1$, so to check that $q_{G_1}$ and $q_{G_2}$ are well defined probabilities, we just need to verify that they are both non-negative. To see that $q_{G_1}$ is non-negative, note that by Lemma~\ref{lem:Vlb}, we have $1/V_{G_1} \ge 1$ and $\pi(G_2)/V_{G_2} \le 1$. Similarly for $q_2$.

Before proving the tree searching strategy guarantees an expected payoff of at least $V_G$ for the Searcher, we illustrate the calculation of the probabilities that define the strategy by considering the tree in Figure~\ref{fig:tree}. The tree searching strategy is specified by the probability $q_{G(A)}$ of searching $G(A)$ before $G(D)$ and the probability $q_{G(B)}$ of searching $G(B)$ before $G(C)$. The first probability is
\[
q_{G(A)} =  \lambda(G) \left(\frac{1}{V_{G(A)}} - \frac{\pi(G(D))}{V_{G(D)}} \right) = \left(\frac{10}{59} \right)\left( \frac{1}{2/3} - \frac{(1/3)(1/2)(3/5)}{1/6} \right) = \frac{9}{59}.
\]
The second probability is
\[
q_{G(B)} = \lambda(G(D)) \left(\frac{1}{V_{G(B)}} - \frac{\pi(G(C))}{V_{G(C)}} \right) = \left( \frac{1}{3} \right) \left( \frac{1}{1/3} - \frac{1/2}{1/2} \right) = \frac{2}{3}.
\]
\begin{lemma} \label{lem:Sbound}
The tree searching strategy $s_G$ ensures an expected payoff of at least $V_G$ against any Hider pure strategy.
\end{lemma}
\begin{proof}
The proof is by induction on the number of vertices of $G$ (it is obviously true for graphs with one vertex). Suppose that $G$ has more than one vertex, and first suppose that $O$ is not a branch vertex, that $O'$ is the vertex adjacent to $O$ and $G'$ is the tree rooted at $O'$. Then by induction, the expected payoff is at least $p_{O'}V_{G'} \equiv V_G$.

Now suppose that $O$ is a branch vertex, and let $G_1$ and $G_2$ be the two branches of $G$, so that depth-first searches of these subtrees are performed in some order.  Suppose, without loss of generality, that the Hider is located in $G_1$. Then, by induction, the expected payoff satisfies
\begin{align*}
P(s) & \ge q_{G_1} V_{G_1} + q_{G_2} \pi(G_2) V_{G_1} \\
& =  \lambda(G) \left(\frac{1}{V_{G_1}} - \frac{\pi(G_2)}{V_{G_2}} \right) V_{G_1} +  \lambda(G) \left(\frac{1}{V_{G_2}} - \frac{\pi(G_1)}{V_{G_1}} \right) \pi(G_2)V_{G_1} \\
& = \lambda(G)(1- \pi(G)) \equiv V_G.
\end{align*}
This completes the proof.
\end{proof}

\begin{theorem}
The value of the game is $V_G \equiv \lambda(G) (1- \pi(G))$. An optimal strategy for the Hider is the tree hiding strategy $h_G$ and an optimal strategy for the Searcher is the tree searching strategy $s(G)$.
\end{theorem}
\begin{proof}
Suppose the Hider uses the strategy $h_G$. Then, by Lemma~\ref{lem:DF}, the expected payoff is at most~$V_G$, so this is an upper bound for the value of the game. 

On the other hand, if the Searcher uses the strategy $s_G$ then by Lemma~\ref{lem:Sbound}, the expected payoff is at least $V_G$, and it follows that the value is at least $V_G$.

Putting together these two bounds on the value, the theorem follows.
\end{proof}

\section{Conclusion}
We have introduced a new search game to model search and rescue operations in which there is a threat the Searcher will be captured herself. We solved the game in which an arbitrary number of targets must be captured, and in the case of one target we solved the game when played on the vertices of a graph. There are many open questions and variations of the game that must be left to future work. What is the solution to the game for $k>1$ on trees, or for $k=1$ on other classes of graphs? What happens if we relax the assumption that the Searcher starts at a fixed vertex, as in \cite{BK13, BK15}? We could also consider a variation of the game in which there is more than one Searcher, similarly to the scheduling problem in~\cite{Agnetis}. This paper scratches the surface, but we believe this opens up an interesting and potential fruitful avenue of new study.

\section*{Acknowledgements}
This material is based upon work supported by the National Science Foundation under Grant No. IIS-1909446.

%%%%%%%%%%%%%%%%%
\end{document}